\documentclass[11pt]{article}
\usepackage{hyperref,natbib,fullpage,amsmath,amssymb,graphicx,amsthm}

\newcommand{\m}[1]{\mbox{\bf{#1}} }
\newcommand{\etr}{ \ensuremath{ {\rm etr}  }   }
\newcommand{\tr}{ \ensuremath{ {\rm tr}  }   }

\renewcommand{\v}[1]{\mbox{\boldmath{${\rm #1}$}}}
\newtheorem{prop}{Proposition}

\newcommand{\E}[1]{{\rm E}[ \ensuremath{ #1 } ]  }

\begin{document}
\title{A hierarchical eigenmodel for pooled covariance estimation}
\author{Peter D. Hoff \thanks{Departments of Statistics and Biostatistics,
University of Washington,
Seattle, Washington 98195-4322.
Web: \href{http://www.stat.washington.edu/hoff/}{\tt http://www.stat.washington.edu/hoff/}. 
This research was partially supported by NSF grant  SES-0631531. 
The author thanks Michael Perlman for helpful discussions. 
        } 
}
\date{ \today }
\maketitle

\begin{abstract}
While a set of covariance matrices corresponding to different populations
are unlikely to be exactly equal they
can still  exhibit a high degree of similarity.
For example, some pairs of variables may be positively
correlated across most groups, while the correlation between
 other pairs  may be  consistently
negative. In such cases much of the similarity across covariance matrices
can be described by similarities in their principal axes, the axes
defined by the eigenvectors of the  covariance matrices.
Estimating the degree of across-population eigenvector heterogeneity
can be helpful for a variety of estimation tasks.
Eigenvector matrices can be pooled to form a central
set of principal axes,
and to the extent that the axes are similar, covariance estimates
for populations having small sample sizes can be stabilized
by shrinking %estimates of their population-specific
their
principal axes towards the across-population center.
To this end, this article develops a hierarchical
model and estimation procedure for
pooling principal axes  across several populations.
The model for the across-group heterogeneity
is based on a matrix-valued antipodally symmetric
Bingham distribution that can flexibly describe
notions of ``center'' and ``spread'' for a population
of orthonormal matrices.

\vspace{.2in}
\noindent {\it Some key words}:
Bayesian inference, copula,
Markov chain Monte Carlo, principal components,
random matrix,  Stiefel manifold.

\end{abstract}

\section{Introduction}
Principal component analysis is a well-established 
procedure for describing the features of 
a covariance  matrix. 
Letting $\m U\Lambda \m U^T$  be the eigenvalue decomposition of the 
covariance matrix of a $p$-dimensional 
random vector $\v y$,  the principal components of $\v y$ are the 
elements of the transformed mean-zero vector $\m U^T(\v y-\E{\v y})$. 
From the orthonormality of $\m U$ it follows that the
elements of the principal 
component vector
% has an expectation of zero and that its 
%elements are uncorrelated, 
are uncorrelated, with  variances equal to the 
diagonal of $\Lambda$. 
Perhaps more importantly, the matrix $\m U$  provides a
natural coordinate system for describing the orientation of the 
multivariate density of $\v y$:  Letting  $\v u_j$ denote the 
$j$th column of $\m U$, 
$\v y$ can be expressed as $\v y -  \E{\v y} = z_1 \v u_1 + \cdots +z_p \v u_p$, 
where
$(z_1,\ldots, z_p)^T$ is a vector of 
uncorrelated mean-zero random variables with diagonal covariance matrix $\Lambda$. 

Often the same set of variables are measured in multiple populations. 
Even if the covariance matrices differ across populations, 
it is natural to expect that they share some common structure, 
such as the correlations between some pairs of variables having  common 
signs across the populations. 
With this situation in mind, 
\citet{flury_1984} developed estimation and testing procedures 
for the ``common principal components'' model, 
in which a set of covariance matrices $\{ \Sigma_1,\ldots, \Sigma_K\}$ 
have common eigenvectors, so that
$\Sigma_j = \m U\Lambda_j \m U^T$ for each $j\in \{1,\ldots, K\}$. 
A number of variations of this model have since appeared:
\citet{flury_1987} and \citet{schott_1991,schott_1999} consider cases 
in which only certain 
columns or subspaces of $\m U$ are shared across populations, 
and \citet{boik_2002} describes a very general model in 
which  eigenspaces can be shared between 
all or some of the populations.

These approaches all assume that certain eigenspaces are either exactly 
equal or completely distinct across a collection of covariances matrices. 
In many cases these two alternatives are too extreme, 
and it may be desirable to recognize situations in which 
eigenvectors are similar but not exactly equal. 
To this end, 
this article develops a hierarchical model to assess heterogeneity 
of principal axes across a set of populations. 
This is accomplished with the aid of a 
probability distribution over the orthogonal group $\mathcal O_p$
%that is appropriate for describing heterogeneity across a set of 
%principal component axes. The distribution 
which can be 
used in a hierarchical model for
sample covariance matrices,  allowing for pooling of covariance information 
and a description of similarities and differences across populations. 
Specifically,
% in the case of normally distributed data the 
%within-population sampling model is that the sample 
%sum of squares matrix $S_k = Y_k ^T(I-\frac{1}{n_k}11^T)Y_k$ 
%has a Wishart$(\Sigma_k,n_k-1)$ distribution. The between-population 
this article develops a sampling model for 
across-population covariance heterogeneity in which 
%sampling model is that
% $\Sigma_k = U_k \Lambda_k U_k^T$ and 
\begin{eqnarray}
  p(\m U|\m A,\m B,\m V )& =& c(\m A,\m B)\etr(\m B \m U^T \m V \m A \m V^T \m U)   \\
\m U_1,\ldots, \m U_K & \sim & \mbox{i.i.d.\ } p(\m U|\m A,\m B,\m V) \nonumber  \\ 
\Sigma_k &= & \m U_k \Lambda_k \m U_k^T,  \nonumber
\end{eqnarray}
where $\m A$ and $\m B$ are diagonal matrices and 
$\m V\in \mathcal O_p$.
The above distribution is a type of generalized Bingham distribution 
 \citep{khatri_mardia_1977, gupta_nagar_2000} that is 
appropriate for modeling  principal component axes.
Section 2 of this article describes some %previously unpublished 
 features of this distribution, in particular 
how $\m A$ and $\m B$ represent the variability of $\{ \m U_1,\ldots, \m U_K\}$
and how $\m V$ represents the mode. 
Parameter estimation is discussed in Section 3, in which a 
Markov chain Monte Carlo algorithm is developed which allows 
for the joint estimation of $\{\m A,\m B, \m V\}$ as well as 
 $\{ (\m U_k,\Lambda_k), k=1,\ldots, K\}$. 
The estimation scheme is illustrated with two example data  analyses 
in Sections 4 and 5. The first dataset, previously analyzed 
by \citet{flury_1984} and \citet{boik_2002} among others, 
involves 
skull measurement data on four populations of voles. 
Model diagnostics and comparisons indicate that
%provides
%a comparison to previous work. 
the proposed hierarchical model represents certain 
features of the observed covariance matrices better than do less 
flexible models. 
%% show some data, or variability of pa's here?
The second example involves survey data from 
different states across the U.S.. 
The number of observations per state varies a great deal, and 
many states  have only  a few observations. 
The example shows how the hierarchical model  %provides stable 
%covariance estimation for such situations by 
shrinks correlation estimates 
towards the across-group center when within-group data is limited. 
Section 6 provides a discussion of the hierarchical model and a 
few of extensions of the approach.

\section{A generalized Bingham distribution}
The eigenvalue decomposition of a positive definite covariance
matrix $\Sigma$ is given by $\Sigma = \m U \Lambda \m U^T$, 
where $\Lambda$ is a diagonal matrix of positive numbers
$(\lambda_1,\ldots, \lambda_p)$
 and 
$\m U$ is an orthonormal matrix, so that $\m U^T\m U=\m U\m U^T=\m I$. 
Writing $\m  U \Lambda \m U^T = \sum_{j=1}^p \lambda_j \v u_j \v u_j^T$, 
we see that multiplication of a column of $\m U$ by -1 does not 
change the value of the covariance matrix, highlighting the fact that
that the columns of $\m U$ represent not directions of variation, 
but axes.  As such, any probability model representing 
variability across a set of principal axes should be antipodally 
symmetric in the columns of $\m U$, meaning that $\m U$ is equal in distribution to 
$\m U\m S$ for any diagonal matrix $\m S$ having  diagonal elements equal to 
plus or minus one, since $\m U$ and $\m U\m S$ represent the same 
axes. 

 \citet{bingham_1974} described a probability distribution having a density
proportional to $\exp\{ \v u^T \m G\v u\}$ 
for normal vectors $\{ \v u: \v u^T\v u=1 \}$. 
This density has 
antipodal symmetry, 
making it a candidate model for a random axis. 
\citet{khatri_mardia_1977} and \citet{gupta_nagar_2000} discuss
a matrix-variate version of the Bingham distribution, 
%\citet{bingham_1974}'s 
\begin{equation}
 p(\m U| \m G,\m H ) \propto \etr( \m H \m U^T \m G \m U ) 
\label{eq:gbf}
\end{equation}
where $\m G$ and $\m H$ are $p\times p$ symmetric matrices. 
Using the eigenvalue decompositions 
$\m G= \m V \m A \m V^T$ and $\m H=\m W \m B\m  W^T$, this density can be rewritten as
$ p(\m U| \m A,\m B,\m V,\m W ) \propto \etr(\m  B [\m  W^T\m  U^T\m  V]\m  A  [\m V^T\m  U\m  W ]  )$. 
A well-known feature of this density is that it 
depends on $\m A$ and $\m B$ only through the differences among their 
diagonal elements. 
This is because for any orthonormal matrix $\m X$ (such as $\m V^T\m  U\m  W$), we have
\begin{eqnarray*}
\tr(  [\m B+d\m I]\m X^T [\m A+c\m I]\m X) &=& \tr(\m B\m  X^T\m A\m X) + d\times \tr(\m X^T\m A\m X)+c \times \tr(\m B\m  X^T \m  X) + 
      cd \times\tr(\m  X^T\m X) \\
  &=& \tr(\m B \m X^T\m A\m X) +  d\times \tr(\m A) + c\times \tr(\m B) + cdp
\end{eqnarray*}
and so the probability densities $p(\m U|\m A,\m B)$ and $p(\m U|\m A+c\m I,\m B+d\m I)$ are
proportional as functions of $\m U$ and therefore
equal.  By convention $\m A$ and $\m B$ are usually taken to be non-negative.
In what follows we will set the smallest eigenvalues $a_p$ and $b_p$
to be equal to zero.

Although a flexible class of distributions for orthonormal matrices, 
 densities of the form (\ref{eq:gbf}) % of  matrix-variate Bingham densities 
are 
not necessarily  antipodally symmetric. However, 
we can identify conditions on $\m G$ and $\m H$ which give the desired symmetry. 
If either $\m G$ or $\m H$ have only one unique eigenvalue, then
 $\tr(\m H \m U^T \m G \m U)$ is constant in $\m U$ and therefore trivially 
antipodally symmetric in the columns of $\m U$. Otherwise, we have the following result:
\begin{prop}
 If  $\m G$ and $\m H$ both have more than one unique 
eigenvalue, then a necessary and sufficient condition for 
  $\tr(\m H \m U^T \m G \m U)$ to be antipodally symmetric in the columns 
 of $\m U$ is that $\m H$ be a diagonal matrix. 
\end{prop}
\begin{proof}  
The
symmetry condition requires that
$\tr(\m H \m U^T \m G \m U ) = \tr( \m S \m H \m S \m U^T \m G \m U )$
%    &=& \tr ( [S H S] X^T G X)  =  \tr ( \tilde H X^T G X). 
%\end{eqnarray*}
for all $\m U\in\mathcal O_p$ and diagonal sign matrices $\m S$.  If $\m H$ is diagonal then $\m S\m H\m S=\m H$ and so
diagonality is a sufficient condition for antipodal symmetry.
To show that it is also a necessary condition, first
let $\m V$ and $\m A$ be the  eigenvector and eigenvalue matrices of $\m G$. 
Then for any orthonormal $\m X$ we have $\m X=\m U^T\m V$ if $\m U=\m V\m X^T$. Therefore 
the symmetry condition is that 
\[ \tr ( [\m H-\m S\m H\m S] \m X \m A \m X^T ) = 0 \  \mbox{ for all $\m X\in \mathcal O_p$ and 
$\m S \in \{ {\rm diag}(\v s), \v s \in \{\pm 1\}^p \}$ }. \]
%diagonal matrices $S$ of signs }.  \]
If we let ${\rm diag}(\m S)=(-1,1,1,\ldots, 1)$  then
$\m D=\m H-\m S\m H\m S$ is zero except for possibly 
$\m D_{[1,-1]}$ and $\m D_{[-1,1]}$, the first row and column of \m D absent 
$d_{1,1}$. 
Additionally, if not all entries are zero then $\m D$ is of rank two with  eigenvectors 
$\v e$ and $\m S\v e$, where $\v e=(|\m H_{[-1,1]}|, h_{2,1},h_{3,1},\ldots, 
    h_{p,1})/(\sqrt {2} |\m H_{[-1,1]}|)$, and corresponding eigenvalues $\pm d$,
where $d= 2|\m H_{[-1,1]}|$. Writing the symmetry condition in terms of the eigenvalues and 
vectors of $\m D$ gives
\[  0=\tr( \m D \m X \m A \m X^T )
   = d \sum_{i=1}^p a_i \v x_i^T \left  ( 
  \v e\v  e^T - \m S\v  e\v  e^T \m S \right ) 
  \v x_i .  \]
The symmetry condition requires that this hold for all orthonormal $\m X$. 
Now if $k$ and $l$ are the indices of  any two unequal eigenvalues of $\m G$, 
we can let $\v x_k=\v e$ and $\v x_l= -\m S\v e$, giving
\begin{eqnarray*} 
 0= \tr( \m D \m X \m A \m X^T ) &=& d  [
    a_k (1- \v e^T\m S\v e\v e^T\m S\v e )  + a_l (\v e^T\m S\v e\v e^T\m S\v e-1) ]  \\
  &=& d (1-[\v e^T\m S\v e ]^2)(a_k-a_l).  % \\
%(e^TSe)^2 &=& 1
\end{eqnarray*}
Since $a_k\neq a_l$ by assumption, this means that either $d=0$
or $(\v e^T\m S\v e)^2=1$. Neither of these conditions are met unless all entries of
$\m D$ are zero, implying that the off-diagonal elements in the first row 
and column of $\m H$ must be zero. Repeating this argument with 
the diagonal of $\m S$ 
ranging over all $p$-vectors consisting of one negative-one and $p-1$ positive ones 
shows that all off-diagonal elements of $\m H$ must be zero. 
\end{proof}

Based on this results we fix the eigenvector matrix of $\m H$ to be $\m I$ and 
our column-wise antipodally symmetric model for
$\m U\in \mathcal O_p$ is
\begin{equation}
 p_B(\m U | \m A,\m B ,\m V )  = c(\m A,\m B) \etr(\m B \m U^T \m V \m A \m V^T\m  U )
\label{eq:asbing}
\end{equation}
where  $\m A$ and $\m B$ are diagonal matrices with $a_1\geq a_2\geq\cdots \geq
 a_p=0$, $b_1\geq b_2\geq\cdots  \geq b_p=0$
and $\m V \in \mathcal O_p$.
Interpreting  these parameters is made easier by 
writing $\m X=\m V^T\m U$ and expanding out the exponent 
of  $p_B$ as
\begin{equation} \tr( \m B \m U^T\m V\m A\m V^T\m U) = \sum_{i=1}^p \sum_{j=1}^p 
   a_i b_j (\v v_i^T\v u_j)^2 =  \sum_{i=1}^p \sum_{j=1}^p a_i b_j x_{i,j}^2 = 
    \v a^T (\m X\circ \m X)\v  b, 
\label{eq:eexp}
\end{equation}
where ``$\circ$'' is the Hadamard product denoting element-wise multiplication.
The value of $x_{i,j}^2$ describes how close column $i$ of $\m V$ is to 
column $j$ of $\m U$. Since both $a$ and $b$ are in decreasing order, 
$a_1b_1$ is the largest term and the density will be large when 
$x_{1,1}^2$ is large. However, due to the orthonormality of $\m X$, a large 
$x_{1,1}^2$  restricts $x_{1,2}^2$ and $x_{2,1}^2$
to be small, which then allows $x_{2,2}^2$ to be large. Continuing on this 
way suggests that the density is maximized if $\m X\circ \m X$ is the identity, 
 i.e.\ $\m U=\m V\m S$ for some diagonal sign  matrix $\m S$. 
%A related result appears without proof in \citet{constantine_muirhead_1976}. 
\begin{prop}
The modes of $p_B$ include $\m V$ and $\{ \m V\m S : \m S = {\rm diag}(\v s), \v s\in \{\pm 1\}^p\}$. 
If the diagonal elements of $\m A$ and $\m B$ are 
distinct then these are the only modes. 
\end{prop}
\begin{proof}
The matrix  % that the diagonal elements of $A$ and $B$ are ordered and
 $\m X\circ \m X$ is an  element of the set of orthostochastic matrices,
 a subset of the  doubly stochastic matrices.
The set of doubly stochastic matrices is a compact convex set
whose extreme points are the permutation matrices. Since every
element of this compact convex set can be written as a convex combination
of the extreme points, we have
\[ \max_{\m X\in \mathcal O_p}  \v a^T(\m X\circ \m X)\v b  \leq 
   \max_{\theta}  \v a^T (\sum \theta_k \m P_k )\v b \]
for probability distributions $\theta$ over the finite set of permutation 
matrices. 
If the elements of 
$\v a$ and $\v b$ are distinct and ordered it is 
easy to show that  $\v a^T \m P \v b$ is uniquely maximized 
over permutation matrices $\m P$ by $\m P=\m I$, and so the right-hand side is maximized when 
$\theta$ is the point-mass measure on $\m I$.  Since $\m I$ is orthostochastic, 
the maximum  on the left-hand side is achieved at $(\m X\circ \m X)=\m I$. 
%. Therefore $a^T Z b$ is maximized in $Z$ over 
%the doubly stochastic matrices by $Z=I$. 
\end{proof}

If two adjacent eigenvalues are equal then the density has 
additional maxima. 
For example, if  $b_1=b_2$ then $\v a^T \m Z \v b$ is maximized over
doubly stochastic matrices  $\m Z$
by any convex combination of the matrices corresponding to the permutations
$\{1,2,3,\ldots, p\}$ and $\{2,1,3,\ldots,  p\}$. 
In terms of $\m U$, this would mean that modes are such that 
  $\v u_i^T\v v_i=\pm 1$ for $i>2$, with $\v u_1$ and $\v u_2$ being any 
orthonormal vectors in the null space of $\{ \v v_3,\ldots, \v v_p$\}. 
More generally, how the parameters $(\m A,\m B)$ control the variability of 
$\m U$ around $\m V$ 
can be seen 
by rewriting the exponent of (\ref{eq:asbing}) a few different ways. 
For example, equation (\ref{eq:eexp}), can be expressed as
\begin{eqnarray}
 \tr(\m B \m U^T\m V\m A\m V^T \m U) 
 &=&  \sum_{i=1}^p \sum_{j=1}^p a_i b_j (\v v_i^T \v u_j)^2  \nonumber \\
 &=& \sum_{j=1}^p b_j\v  u_j^T (\m V\m A\m V^T)\v  u_j  
\label{eq:eexp2}
\end{eqnarray}
%\begin{eqnarray*}
%\tr(B U^TVAV^T U) &=& \sum_{i=1}^p \sum_{j=1}^p a_i b_j (v_i^Tu_j)^2 \\
% &=& \sum_{i=1}^p \sum_{j=1}^p  b_j u_j^T( a_i v_i v_i^T) u_j  \\
% &=& \sum_{j=1}^p  b_j u_j^T\left ( \sum_{i=1}^p a_i v_i v_i^T \right ) u_j  
%\end{eqnarray*}
From equations (\ref{eq:eexp}) and (\ref{eq:eexp2}) it is clear that
$b_j=b_{j+1}$ implies that 
$\v u_j \stackrel{d}{=} \v u_{j+1}$ and 
$a_{i}=a_{i+1}$ implies $\v v_i^T\v u_j\stackrel{d}{=} \v v_{i+1}^T\v  u_j$. 
%$a_i=a_{i+1}$ implies that  
% $BU^Tv_i \stackrel{d}{=} BU^T v_{i+1}$. 
In this way the model can represent \emph{eigenspaces} of high probability, 
not only eigenvectors, providing a probabilistic analog 
to the common space models of \citet{flury_1987}. 
%In terms of the 
%Bingham density in $X$, this means that modes of 
%$p(X|A,B,V,W)$ are matrices whose columns are equal to those of
%$VSW^T$ for  diagonal matrices $S$ with diagonal elements of $\pm 1$. 
%% need to assume a, b positive
To illustrate this further, Figure \ref{fig:xxplots} shows the 
expectations of the squared elements of $\m X = \m V^T\m U$ for two different values
of $(\m A,\m B)$ (calculations were based on a Monte Carlo approximation 
scheme described in \citet{hoff_2007c}). 
The plot in the first panel is based on the generalized Bingham distribution 
in which 
%${\rm diag}(A) = {\rm diag}(B)  = (6,3,0,0,0) $. 
${\rm diag}(\m A) = {\rm diag}(\m B)  = (7,5,3,0,0,0) $. 
For this distribution, 
$\v u_1$, $\v u_2$ and $\v u_3$ are highly concentrated around $\v v_1$, $\v v_2$ and 
  $\v v_3$  respectively. Since $b_4=b_5=b_6$,  the vectors 
   $\v u_4$, $\v u_5$ and $\v u_6$  are equal in distribution and close to 
 being uniformly distributed on the null space of $(\v v_1,\v v_2,\v v_3)$. 
These particular values of $(\m A,\m B)$ could represent a situation 
in which the first three eigenvectors are conserved across populations but the 
others are not.
%with the remaining
%eigenvectors representing isotropic variation. 
The second panel of Figure \ref{fig:xxplots} represents a more 
complex situation in which  ${\rm diag}(\m A) = (7,5,3,0,0,0)$ and 
${\rm diag}(\m B)  = (7,7,0,0,0,0) $. For these parameter values
the following components of $\m X=\m V^T\m U$ are equal in distribution:
 columns 1 and 2;
columns 3, 4, 5 and 6; 
rows 2 and 3; rows
4, 5 and 6. 
Such a distribution might represent a situation in which 
a vector $\v v_1$ is shared across populations, but it is equally 
likely to be represented within a population by 
either $\v u_1$ or $\v u_2$.

\begin{figure}
\centerline{\includegraphics[height=2.5in]{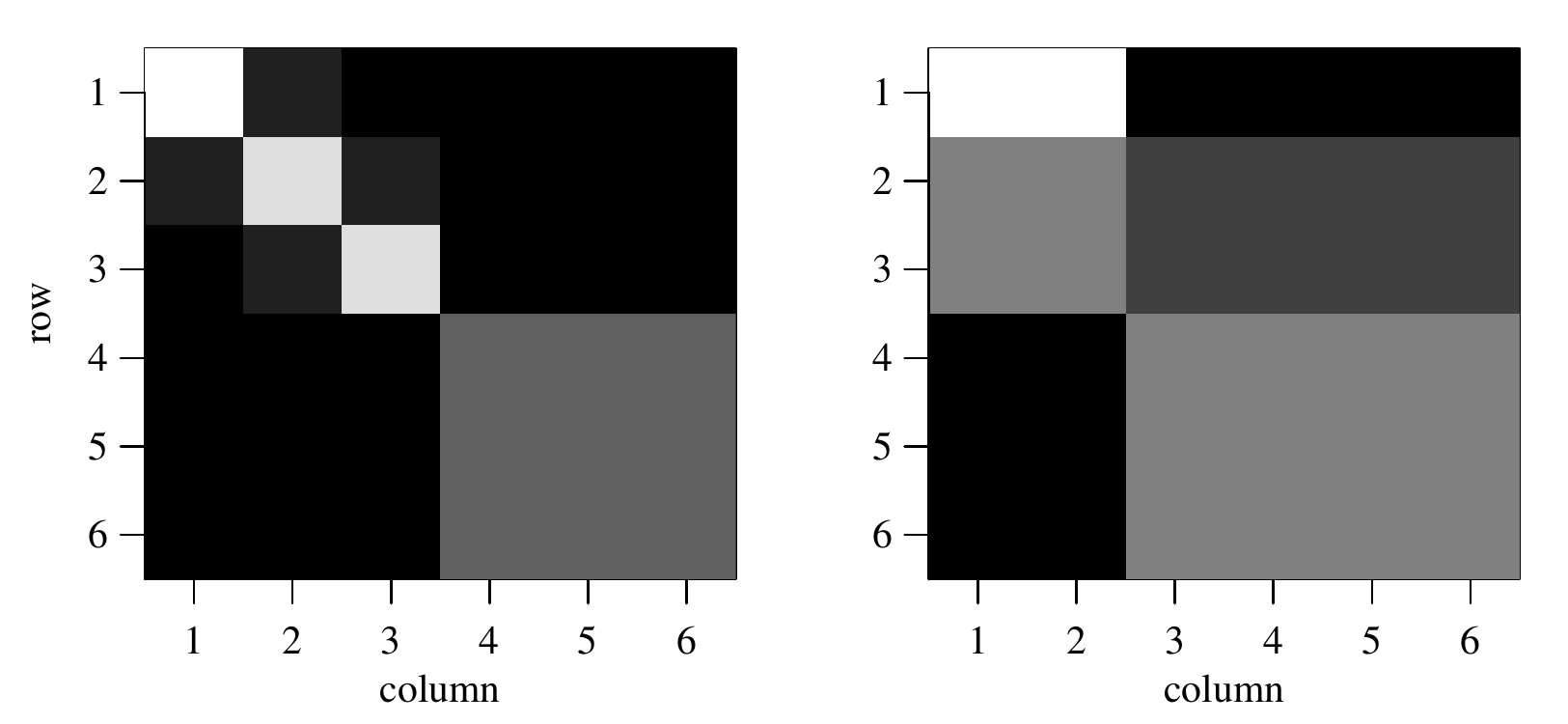}}
\caption{Expected values of the squared entries of $\m X=\m V^T\m U$ under two 
different Bingham distributions. Light shading indicates 
high values. }
\label{fig:xxplots}
\end{figure}

\section{Pooled estimation of covariance eigenstructure}
In the case of normally distributed data the 
sampling model for  $\m S_k = \m Y_k ^T(\m I-\frac{1}{n_k}\v 1\v 1^T)\m Y_k$, 
the   observed
sum of squares matrix in population $k$, 
is a Wishart distribution. 
Combining this  with the model for  principal axes 
 developed in the last section gives the following 
hierarchical  model for covariance structure:
\begin{eqnarray*}
\m U_1,\ldots, \m U_K &\sim & {\rm i.i.d.}\  p(\m U|\m A,\m B,\m V)  \ \ \ \ \ \  \ \ \  \ \  \ \mbox{(across-population variability)}  \\
\m S_k &\sim& {\rm Wishart}(\m U_k \Lambda_k \m U_k^T,n_k-1)
  \ \ \mbox{(within-population variability)}
\end{eqnarray*}
The unknown parameters  to estimate include $\{\m A,\m B,\m V\}$ 
as well as the within-population
covariance matrices, parameterized as $\{ (\m U_1,\Lambda_1),\ldots, 
 (\m U_K,\Lambda_K) \}$. 
In this section we describe a Markov chain Monte Carlo algorithm 
that generates approximate samples from the posterior distribution for 
these parameters, allowing for estimation and inference. 
The Markov chain is constructed with Gibbs sampling, in which 
each parameter is 
iteratively 
resampled from its 
 full conditional 
distributions. 
%This distribution of these sampled parameter values
%converges to the posterior 
%where available, and Metropolis-Hastings proposals 
%otherwise. 
We first describe conditional updates
of the across-population parameters $\{\m A,\m B,\m V\}$, 
then describe 
pooled estimation of each $\m U_k$, which combines the population-specific 
information $\m S_k$ with the across-population information in $\{\m A,\m B,\m V\}$. 
%Additionally, we discuss different models for pooling the eigenvalue matrices
%$\{\Lambda_1,\ldots \Lambda_K\}$. 
%A summary of the Markov chain Monte Carlo algorithm is outlined 
%at the end of the section. 

\subsection{Estimation of across-population parameters}
%\paragraph{Full conditional distribution of $V$}
Letting the prior distribution for $\m V$ be the uniform (invariant) 
measure on $\mathcal O_p$, we have
\begin{eqnarray*}
  p(\m V|\m A,\m B,\m U_1,\ldots, \m U_K) &\propto& p(\m V) \prod_{k=1}^K p(\m U_k|\m A,\m B,\m V)    \\
   &\propto& \etr( \sum_{k=1}^K  \m B \m U_k^T \m V \m A \m V^T \m U_k  ) \\
   &=&  \etr( \sum_{k=1}^K  \m  A \m V^T \m U_k \m  B \m  U_k^T \m V ) 
   = \etr( \m A \m V^T  [\sum \m U_k \m B \m U_k^T ] \m V ), 
\end{eqnarray*}
and so the full conditional distribution of $\m V$ is a generalized Bingham
distribution, of the same form as described in the previous 
section. Conditional on the values $\{\m A,\m B,\m U_1,\ldots,\m U_K\}$, 
pairs of columns of $\m V$ can be sampled from their full 
conditional distributions using a method described in 
\citet{hoff_2007c}.  

Obtaining full conditional  distributions for $\m A$ and $\m B$ 
is more complicated. 
The joint density of $\{\m U_1,\ldots, \m U_K\}$
is 
$c(\m A,\m B)^K \etr ( \sum \m A \m V^T  \m U_k \m B \m U_k^T \m V )$.
From the terms in the exponent  we have 
\[ \tr (\sum \m B \m U_k^T\m V^T \m A \m V^T\m  U_k)= \sum_{k=1}^K \tr(\m B \m U_k^T\m V \m A \m V^T\m U_k)\\
%  &=&  \sum_{k=1}^K \sum_{i=1}^p \sum_{j=1}^p a_i b_j (v_i^Tu_{j,k})^2  \\
   = \sum_{k=1}^K  \sum_{i=1}^p \sum_{j=1}^p a_i b_j  (\v v_i^T\v u_{j,k})^2  = \v a^T \m M\v b 
\]
where $\m M$ is the matrix $\m M=\sum_{k=1}^K (\m V^T \m U_k)\circ (\m V^T \m U_k)$.  
% with entry $i,j$ equal to 
%   $\sum_{k=1}^K (v_i^Tu_{j,k})^2$. 
 The normalizing constant $c(\m A,\m B)$ is equal to $_0F_0(\m A,\m B)^{-1}$, 
where  $_0F_0(\m A,\m B)$ is a
type of hypergeometric function with matrix arguments \citep{herz_1955}. 
Exact calculation 
of this quantity is problematic, although approximations have been 
discussed in \citet{anderson_1965,constantine_muirhead_1976} and \citet{muirhead_1978}.
% \citet{constantine_muirhead_1976} and \citet{muirhead_1978}. 
The first-order term in these approximations is 
\[ c(\m A,\m B) \approx  \tilde c(\m A,\m B) =2^{-p} \pi^{- {p \choose 2}} e^{-\v a^T\v b} \prod_{i<j}
   (a_i-a_j)^{1/2}(b_i-b_j)^{1/2} . \]
This gives the following approximation to the 
likelihood  for $\m A,\m B$: 
\begin{eqnarray}
  p(\m U_1,\ldots, \m U_K|\m A,\m B,\m V)
 &\approx & \tilde c(\m A,\m B)^K \etr(\v a^T \m M \v b) \nonumber  \\
 &\propto&  \exp\{ - \v a^T  (K\m I-\m M)\v b\} \prod_{i<j}
   (a_i-a_j)^{K/2}(b_i-b_j)^{K/2} 
\label{eq:ablik}
\end{eqnarray}
However, there is an identifiability issue with this likelihood:
As seen above, since $\m A$ and $\m B$ are diagonal, 
$\tr ( \m B \m X^T\m  A \m X ) $ simplifies to
$\v a^T (\m X \circ \m X)\v b= \sum_i \sum_j a_i b_j x_{i,j}^2 $.
This means that for any $c>0$, $p(\m U|\m A,\m B,\m V) = p(\m U| c\m A,c^{-1}\m B,\m V)$, and so
the scale of $\m A$ and $\m B$ are not separately identifiable.
To account for this, we parameterize $\m A$ and $\m B$ as follows:
\begin{eqnarray*}
{\rm diag}(\m A) &=& (a_1,\ldots, a_p) = \sqrt{w}(\alpha_1,\ldots, \alpha_p) \\
{\rm diag}(\m B) &=& (b_1,\ldots, b_p) = \sqrt{w}(\beta_1,\ldots, \beta_p) 
\end{eqnarray*}
where $w>0$, $1=\alpha_1>\alpha_2>\cdots >\alpha_{p-1}> \alpha_p =0$ and
 $1=\beta_1>\beta_2> \cdots > \beta_{p-1} >\beta_p =0$.
Rewriting (\ref{eq:ablik}) in terms of these parameters and multiplying 
by a prior distribution $p(\v \alpha,\v \beta,w) $ gives
\begin{equation}
  p(\v\alpha,\v \beta,w| \m M) \propto 
  p(\v \alpha,\v \beta,w) \times \exp\{ -w \v \alpha^T  (K\m I-\m M)\v \beta \} 
   w^{  {p \choose 2}  K/2 }
   \prod_{i<j}
   (\alpha_i-\alpha_j)^{K/2}(\beta_i-\beta_j)^{K/2}. 
\label{eq:abwpost}
\end{equation}
In what follows we take the prior distribution $p(\v \alpha,\v \beta,w)$ 
such that $1>\alpha_{2}>\cdots > \alpha_{p-1} >0$ 
and  $1>\beta_{2}>\cdots > \beta_{p-1} >0$ are two independent 
sets of  order statistics of uniform random variables on $[0,1]$, 
and $w$ has a gamma distribution. With these priors, 
the values of $\v \alpha$ and $\v \beta$ can be sampled from their 
full conditional distributions on a grid of $[0,1]$, and 
the full conditional distribution of $w$ is  a gamma distribution.
For example, if $w\sim$gamma($\eta_0/2,\tau_0^2/2$) {\it a priori}, then 
 $p(w|\v \alpha, \v \beta, \m M) $ is 
gamma$(\eta_0/2+{p \choose 2}K/2  , \eta_0 \tau_0^2/2+\v \alpha^T(K\m I-\m M)\v \beta )$.

We should keep in mind that this full conditional distribution is based
on an approximation to the normalizing constant $c(\m A,\m B)$
(although it is a ``bona fide'' full conditional distribution 
under the prior $\tilde p(\m A,\m B) =
 p(\m A,\m B)[\tilde c(\m A,\m B)/c(\m A,\m B)]^K$). 
Results from 
\citet{anderson_1965} show that 
in terms of the parameters 
$\{\v \alpha,\v \beta, w\}$, 
\begin{equation}
  \frac{c(\m A,\m B)}{\tilde c(\m A,\m B)} \approx  
    1+ \frac{1}{4w}\sum_{i<j} [(\alpha_i-\alpha_j)(\beta_i-\beta_j)]^{-1} +
    O(\frac{1}{w^2})
\label{eq:capprox}
\end{equation}
Further terms in the expansion are  available in \citet{anderson_1965}. 
In problems where $w$ is large then the approximation is likely to be 
a good one. 
In cases where the differences between consecutive $\alpha_i$'s 
or $\beta_j$'s is small compared to $1/w$ then it may be desirable to 
correct for the approximation using a few additional terms from 
(\ref{eq:capprox}) via a Metropolis-Hastings procedure. 
For example, letting $h(\v \alpha,\v \beta,w)= 1+
 \sum_{i<j} [4w (\alpha_i-\alpha_j)(\beta_i-\beta_j)]^{-1}$, 
if $w_s$ is the current value of $w$ in the Markov chain and 
$\tilde w$ is sampled from the approximate  full conditional distribution 
of $w$ based on (\ref{eq:abwpost}), then the correction can be implemented
as follows:
\begin{enumerate}
\item sample a proposal $\tilde w$ from the full conditional of $w$ 
based on  (\ref{eq:abwpost}); 
\item sample $u\sim$ uniform(0,1) and compute 
      $r=  [h(\v \alpha,\v \beta,\tilde w)/h(\v \alpha,\v \beta,w_s)]^K$; 
\item if $u<r$ then set $w_{s+1}=\tilde w$, otherwise set 
       $w_{s+1}= w_s$. 
\end{enumerate}
A similar procedure using one more order in the expansion
%This particular procedure 
is used in the example data analyses of the 
next section. 

Finally, we note that there has been some recent progress
in computing $_0F_0(\m A,\m B)$ exactly. 
\citet{koev_edelman_2006} provide an algorithm 
that is fast enough to be used in MCMC algorithms for 
problems in which $p$ roughly 5 or less and the 
values of 
$\m A$ and $\m B$ are not too large. For other problems the 
approximations based on (\ref{eq:abwpost}) and 
(\ref{eq:capprox}) still seem necessary.

%To summarize, posterior inference for $(A,B,V)$ given $U_1,\ldots, U_K$
%proceeds via 
%Given current values of $\{ A,B,V\}^s$, construct new values
%   $\{ A,B,V\}^{s+1}$ as follows:
%\begin{enumerate}
%\item Update $(A,B)$:
%\begin{enumerate}
%\item generate  $(a^*,b^*)$ from a symmetric proposal distribution 
%\item sample $z\sim$ uniform(0,1)  and set 
%\[    (A,B)^{s+1} = \left  \{ \begin{array}{l}  
%      (A^*,B^*)\  {\rm if} \ z< f(a^*,b ^*)/f(a,b), \\
%      (A,B)^{s}\  {\rm if} \ z> f(a^*,b ^*)/f(a,b). \end{array} \right . \]
%\end{enumerate}
%\item Update $V$:
%\begin{enumerate}
%\item  Compute $M=  \sum_{k=1}^K U_k B U_k^T $; 
%\item Iteratively over pairs of columns $(j_1,j_2)$ of $V$, 
%    sample $\{ v_{j_1},v_{j_2} \}$ from the generalized 
%    Bingham$(A,M)$ distribution  
%    conditional on $\{ v_j, j\not\in \{j_1,j_2\}$. 
%\end{enumerate}
%\end{enumerate}

\subsection{Estimation of population-specific principal axes}
As described above, the within-population sampling  model for 
the sample sum-of-squares matrix $\m S_k$ is Wishart$(\m U_k\Lambda_k \m U_k^T,n_k-1)$, 
so that as a function of $\m S_k$, $\m U_k$ and $\Lambda_k$, 
\begin{equation}
p( \m S_k |\m U_k, \Lambda_k) \propto    |\Lambda_k|^{-(n_k-1)/2} | \m S_k|^{(n_k-p-2)/2} \times \etr(-\frac{1}{2} \Lambda_{k}^{-1} \m U_k^T \m S_k \m U_k  ). 
\label{eq:wishart}
\end{equation}
In the absence of information from other populations, a uniform 
prior distribution on $\m U_k$ would yield a generalized Bingham 
full conditional distribution for $\m U_k$. However, combining 
(\ref{eq:wishart}) with the across-population information 
$p(\m U_k|\m A,\m B,\m V)$ gives a non-standard full conditional distribution
for $\m U_k$:
\begin{eqnarray*}
p(\m U_k|\m S_k,\Lambda_k,\m A,\m B,\m V) &\propto & 
     p( \m S_k |\m U_k, \Lambda_k) \times  p(\m U_k|\m A,\m B,\m V) \\
&\propto&  \etr(-\frac{1}{2} \Lambda_{k}^{-1} \m U_k^T \m S_k\m  U_k  ) \times 
           \etr(\m B\m U_k^T\m V\m A\m V^T\m U_k). 
\end{eqnarray*}
The terms in the exponents are difficult to combine as they are 
both quadratic in $\m U_k$. Writing out the expression in terms of 
the columns $\{ \v u_{1,k},\ldots, \v u_{p,k}\}$ yields some insight:
\begin{eqnarray*}
\tr(\m B\m U_k^T\m V\m A\m V^T\m U_k -\frac{1}{2} \Lambda_{k}^{-1} \m U_k^T \m S_k \m U_k  ) &=& 
 \sum_{j=1}^p \v u_{j,k}^T \left ( b_j \m V \m A \m V^T -\frac{1}{2}\lambda_{j,k}^{-1} \m S_k
   \right ) \v u_{j,k} 
\end{eqnarray*}
This suggests that the full conditional distribution of the 
$j$th column vector of $\m U_k$ is a vector-valued Bingham distribution. 
This is true in a very limited sense: Since $\m U_k$ is an orthonormal 
matrix the full conditional distribution of $\v u_{j,k}$ given 
the other columns of $\m U_k$ must have support only on $\pm \tilde {\v u}$, 
where $\tilde{\v u}$ represents the null space of the vectors
 $\{ \v u_{1,k},\ldots, \v u_{j-1,k},\v u_{j+1,k},\ldots, \v u_{p,k}\}$. 
Iteratively sampling the columns of $\m U_k$ from their full conditional 
distributions would therefore produce a reducible  Markov chain which would 
not converge to the target posterior distribution. 
One remedy to this situation, used by \citet{hoff_2007c} in the 
context of sampling from the Bingham distribution, 
 is to sample from the full conditional 
distribution of columns taken two at a time. 
Conditional on $\{ \v u_{3,k},\ldots, \v u_{p,k}\}$, the vectors 
$\{ \v u_{1,k},\v u_{2,k}\}$ are equal in distribution to 
  $\m N \m Z$, where $\m N$ is any $p\times 2$ dimensional 
 orthonormal basis for the null 
space of $\{ \v u_{3,k},\ldots, \v u_{p,k}\}$ and $\m Z$ is a random 
  $2\times 2$ orthonormal matrix whose density with respect to the 
uniform measure is proportional to 
\[ 
p(\m Z) \propto \etr\left ( \v z_1^T  \m G \v z_1 + \v z_2^T \m H \v  z_2 \right ), 
 \]
where $\m G= \m N^T(b_1 \m V\m A\m V^T-\lambda_{1,k}^{-1} \m S_k)\m N$, 
      $\m H= \m N^T (b_2 \m V\m A\m V^T-\lambda_{2,k}^{-1} \m S_k)\m N$ and 
$\v z_1$ and $\v z_2$ are the columns of $\m Z$. 
Since $\m Z$ is orthogonal, we can  parameterize it as
%\[ \m Z = \left ( \begin{array}{rr} \cos \phi & s \sin \phi \\
%                          - \sin \phi & s \cos \phi \end{array} \right ) \]
\[ \m Z = \left ( \begin{array}{rr} \cos \phi & s \sin \phi \\
                           \sin \phi & -s \cos \phi \end{array} \right ) \]
for some $\phi \in (0,2\pi)$ and $s=\pm 1$.
%The second column $Z_{[,2]}$ of $Z$ is a linear function of
%the first column $Z_{[,1]}$,
%and 
The uniform density on the circle is constant in $\phi$,
so the joint density of
 $(\phi, s)$ is simply $p(\m Z(\phi,s))$.
Sampling from this distribution can be accomplished by first sampling
$\phi \in (0,2\pi)$ from a density proportional to
\[ p(\phi )\propto 
  \exp( [g_{1,1} + h_{2,2}] \cos^2 \phi  +
         [h_{1,1} + g_{2,2}] \sin^2 \phi  +
        [g_{1,2}+g_{2,1} -h_{1,2}-h_{2,1}]\cos \phi \sin \phi ) ,\]
and then sampling $s$  uniformly from $\{-1,+1\}$.

\subsection{Estimation of eigenvalues}
From (\ref{eq:wishart}) we see that the conditional distribution of $\Lambda_k$ 
given $\m U_k$ and  $\m S_k$ has the following form:
\begin{eqnarray*}
p(\Lambda_k|\m U_k,\m S_k) &\propto&  p(\Lambda_k) |\Lambda_k|^{-(n_k-1)/2}\etr(-\frac{1}{2}\Lambda_k^{-1} \m U_k^T \m S_k \m U_k) \\ 
&=& p(\Lambda_k) \prod_{j=1}^p \lambda_{j,k}^{-(n_k-1)/2} \exp\{ -\frac{1}{2}\sum_{j=1}^p 
        \lambda_{j,k}^{-1} \v u_{j,k}^T\m S_k \v u_{j,k} \} 
\end{eqnarray*}
The part not involving the prior distribution has the form of an 
inverse-gamma density, and indeed, if $p(\Lambda_k)$ were the product 
of inverse-gamma densities with parameters $(\nu_0/2, \nu_0\sigma_0^2/2)$ 
then the full conditional distribution of $\lambda_{j,k}$ would be 
inverse-gamma$[ (\nu_0+n-1)/2, (\nu_0\sigma_0^2 + \v u_{j,k}^T\m S_k \v u_{j,k})/2]$. 
However, it may be desirable to add more structure 
to the estimation of the eigenvalues. 
%One modeling decision to make is whether or not the order of the principal axes
%matters.  
In usual one-sample principal component analysis the eigenvalues are 
labeled in order of decreasing magnitude
and attention is  focused on 
the ``first few'' eigenvectors, i.e.\ those corresponding to the 
largest eigenvalues. 
In terms of making comparisons of eigenvectors 
across groups, 
restricting the eigenvalues to be ordered means that the ordered 
columns of  $\m V$ refer to the ordered columns of $\m U$. 
One concern about such a restriction would be how it might affect 
inference in the case of  a shared                        
eigenvector  that is the first principal axes in some groups, and 
possibly the second or third in other groups. 
This sort  of heterogeneity can in fact be represented with the 
generalized Bingham distribution even if the eigenvalues are order-restricted. 
For example, the distribution with ${\rm diag}(\m A) = (a,0,0,\ldots)$ and
${\rm diag}(\m B) = (b,b,0,\ldots)$ represents a population in
which, with equal frequency,
 one of the first two columns of $\m U$ is near the first
column of $\m V$.
%In contrast, the common principal component methods of \citet{flury_1987}
%and others generally allow the 
%eigenvalues of a group-specific covariance matrix to be unordered, as these
%approaches do not have 
Because of this flexibility, in what follows we estimate the eigenvalues in each group as being 
ordered.  A convenient prior distribution is that $p(\Lambda_k)$ is the product 
of inverse-gamma densities described above, but restricted to the space 
$\lambda_{1,k}>\lambda_{2,k} >\cdots > \lambda_{p,k}$.
The full conditional distribution of $\lambda_{j,k}$ is then 
inverse-gamma$[ (\nu_0+n-1)/2, (\nu_0\sigma_0^2 + \v u_{j,k}^T\m S_k \v u_{j,k})/2]$
but restricted to the interval $(\lambda_{j+1,k}, \lambda_{j-1,k} )$. 

% use groups, not populations?

\subsection{Summary of MCMC algorithm}
The unknown parameters in the hierarchical model  %described above 
are the group-specific eigenvectors and values $\{\m U_1,\Lambda_1\},\ldots, 
 \{\m U_K,\Lambda_K \}$ and the parameters $\{\m A,\m B,\m V\}$ describing 
the across-group
heterogeneity of eigenvector matrices. 
The diagonal matrices $\m A$ and $\m B$ are parameterized as
\begin{eqnarray*}
{\rm diag}(\m A) &=& (a_1,\ldots, a_p) = {\sqrt w} (\alpha_1,\ldots, \alpha_p) \\
{\rm diag}(\m B) &=& (b_1,\ldots, b_p) = {\sqrt w} (\beta_1,\ldots, \beta_p)
\end{eqnarray*} 
with $1=\alpha_1>\cdots>\alpha_p=0$ and $1=\beta_1>\cdots>\beta_p=0$. 
Convenient prior distributions are 
$\m V\sim $ uniform $\mathcal O_p$, 
$(\alpha_2,\ldots,\alpha_{p-1})$ and 
$(\beta_2,\ldots,\beta_{p-1})$ are uniform on $[0,1]$ subject to the 
ordering restriction, 
 $w\sim $ gamma$(\eta_0/2,\tau_0^2/2)$
and 
 ($1/\lambda_{1,k},\ldots, 1/\lambda_{p,k})$ 
are the order statistics of a sample from a 
  gamma$(\nu_0/2,\sigma_0^2/2)$ distribution. 
With these prior distributions, a Markov chain 
in the unknown parameters that converges to the posterior distribution
 $p( \{ \m U_1,\Lambda_1\},\ldots, \{\m U_K, \Lambda_K\},\m A,\m B,\m V| \m Y_1,\ldots, \m Y_k)$
 can be constructed by iteration of the following sampling scheme:
\begin{enumerate}
  \item Update the within-group parameters:
  \begin{enumerate}
    \item Update $\{ \m U_1,\ldots, \m U_K\}$: For each $k$ and a 
         randomly selected pair $\{j_1,j_2\} \subset \{ 1,\ldots, p\}$; 
    \begin{enumerate}
      \item let $\m N$ be the null space of the columns 
$\{\v u_{j,k}:j\not \in\{j_1,j_2\} \}$; 
     \item compute  
         $\m G= \m N^T (b_{j_1} \m V\m A\m V^T-\lambda_{j_1,k}^{-1} \m S_k)\m N$,
         $\m H=\m N^T (b_{j_2} \m V\m A\m V^T-\lambda_{j_2,k}^{-1} \m S_k)\m N$; 
     \item  sample $\m Z=(\v z_1,\v z_2)\in \mathcal O_2$ 
       from the density proportional to 
            $\exp( \v z_1^T \m G \v z_1 +  \v z_2^T \m H \v  z_2 )$
 \item set $\v u_{j_1,k}$ to be the first column of $\m N\m Z$ 
and  $\v u_{j_2,k}$ to be the second. 
\end{enumerate}
  \item Update $\{\Lambda_1,\ldots, \Lambda_K\}$: Iteratively for each $j\in \{1,\ldots,p\}$ and $k\in \{1,\ldots,K\}$,
   sample $\lambda_{j,k} \sim $inverse-gamma$[ (\nu_0+n_k-1)/2,
    (\nu_0\sigma_0^2 + \v u_{j,k}^T \m S_k \v u_{j,k})/2 ]$, but constrained to
    be in $(\lambda_{j-1,k},\lambda_{j+1,k})$. 
\end{enumerate}
\item Update the across-group parameters:
\begin{enumerate}
\item  Update $\m V$: Sample $\m V$ from the Bingham density
       proportional to $\etr(\m A \m V^T [\sum \m U_k \m B \m U_k^T] \m V) $. 
\item Update $w$, $\v \alpha$ and $\v \beta$: 
 Compute $\m M=\sum_{k=1}^K (\m V^T \m U_k)\circ(\m V^T\m U_k)$  and
\begin{enumerate}
  \item 
 sample $w\sim {\rm gamma}(\eta_0/2+{p \choose 2}K/2  , \eta_0 \tau_0^2/2+\v \alpha^T[K\m I-\m M]\v \beta )$
   \item for each $i\in \{2,\ldots, p-1\}$ sample 
  $\alpha_i \in( \alpha_{i-1},\alpha_{i+1})$ from the density proportional 
to  
$\exp\{ -\alpha_i (w \v \beta^T \m M_{[i,]})\} \prod_{j:j\neq i} |a_i-a_j|^{K/2}$. 
  \item for each $j\in \{2,\ldots, p-1\}$ sample
  $\beta_j \in( \beta_{j-1},\beta_{j+1})$ from the density proportional 
to  $\exp\{ -\beta_j (w \m M_{[,j]}^T\v \alpha)\} \prod_{i:i\neq j} |\beta_j-\beta_i|^{K/2}$.
\end{enumerate}
\end{enumerate}
\end{enumerate}
As discussed in section 3.1, it may be desirable to make 
Metropolis-Hastings adjustments to 
the steps in 2(b) to account for the approximation to the normalizing constant
$c(\m A,\m B)$. 
Functions and example code for this algorithm,  written in the {\sf R} programming environment,
are available at my website, 
 \href{http://www.stat.washington.edu/hoff/}{\tt http://www.stat.washington.edu/hoff/}.

\section{Example: Vole measurements}
\citet{flury_1987}  describes an analysis of skull measurements
on four different groups of voles. The four groups, defined by 
species and sex, are  male and female {\it Microtus californicus} and 
   male and female {\it Microtus ochrogaster}, 
   having sample sizes 
of 82, 70, 58 and 54 respectively. 
Flury provides
 the sample covariance matrices of four log-transformed 
measurements corresponding to  skull length, toothrow length,
cheekbone width and interorbital width. 
The eigenvectors of the four empirical covariance matrices are
given in Table \ref{tab:eev}.
The first eigenvector in each group can roughly be interpreted
as measuring overall size variation, and its values seem fairly similar across
groups. The remaining eigenvectors also display a high degree of similarity
across  groups.
By  performing statistical  tests  of various hypotheses 
regarding the four population covariance matrices, Flury concludes that 
although the sample covariance matrices appear similar, there is 
enough evidence to reject exact equality of the population matrices. 
Furthermore, Flury rejects a hypotheses that the population covariances 
are proportional to each other, and then suggests a model in which 
the the covariance matrices share a single eigenvector 
(interpreted as corresponding 
to size), with the 
remaining eigenvectors and all of the eigenvalues being distinct 
across groups. 
In this section we reanalyze these data 
using the hierarchical eigenmodel discussed above, 
and compare it to the model in \citet{flury_1987} 
and a few others. In particular, we show that allowing
information-sharing across the groups where appropriate, but not 
forcing any of the eigenvectors to be exactly equal, results in 
a model that better represents features of the observed dataset. 

\begin{table}
\begin{center}
\begin{tabular}{rrrrrrrrrr}
\multicolumn{10}{c}{ {\it M. californicus} } \\
\multicolumn{4}{c}{males}  & & & \multicolumn{4}{c}{females} \\
 36.31 & 27.01&  8.05&  2.78  &&& 52.44&  21.14 & 3.75 & 3.17 \\ 
\cline{1-4} \cline{7-10}
0.49  &  -0.31  &  -0.19  &  0.79 &&  &  0.53  &  -0.31  &  -0.29  &  0.74  \\
0.60  &  -0.10  &  0.76  &  -0.24 &&  &  0.56  &  -0.06  &  0.82  &  -0.11  \\
0.55  &  -0.12  &  -0.63  &  -0.54&&   &  0.57  &  -0.14  &  -0.48  &  -0.65  \\
0.30  &  0.94  &  -0.06  &  0.17 & & &  0.30  &  0.94  &  -0.11  &  0.14  \\ \\
\multicolumn{10}{c}{ {\it M. ochrogaster} } \\
\multicolumn{4}{c}{males}  & & & \multicolumn{4}{c}{females} \\
 36.30 & 9.67& 7.97& 2.80 &&&  35.61& 12.35& 8.32& 3.38 \\
\cline{1-4} \cline{7-10}
0.58  &  0.04  &  -0.38  &  0.71  &&&  0.56  &  0.06  &  0.05  &  0.82  \\
0.45  &  0.72  &  0.51  &  -0.13  & && 0.47  &  0.02  &  0.80  &  -0.38  \\
0.51  &  -0.08  &  -0.51  &  -0.69  &&&  0.66  &  -0.25  &  -0.58  &  -0.40  \\
0.45  &  -0.68  &  0.58  &  -0.02  &&&  0.13  &  0.97  &  -0.17  &  -0.14
\end{tabular}
\end{center}
\caption{Eigenvalues and eigenvectors of empirical covariance matrices. 
    Eigenvalues are given in the first row for each group.}
\label{tab:eev}
\end{table}

Using the sample sizes and sample covariance matrices 
provided in \citet{flury_1987},  centered versions
of $\m Y_k^T\m Y_k$ for each group $k\in \{1,\ldots, 4\}$ were reconstructed and 
used as data 
%to generate  10,000 Markov chain iterations 
for the model  described in Section 3. 
The prior distribution of $w$ was taken
to be a diffuse exponential with a mean of 1000, 
and the prior distribution for the inverse-eigenvalues 
was exponential with a mean of 
1. 
A Markov chain consisting of 10,000 iterations was constructed, for which
parameter values were saved every 10th iteration giving a total
of 1000 posterior samples for each  parameter. 
Mixing of the Markov chains was monitored via a variety of
parameter summaries computed at each saved iteration.
For example, for each saved value of $\m A$  and $\m B$
the average and
standard deviation of the logs of the $(p-1)\times (p-1)=9$
non-zero values of $\m A\circ \m B$ were obtained and 
 plotted sequentially in
the first panel of Figure \ref{fig:voleana}.

A  posterior point
estimate  of $\m V$ can be obtained from the eigenvector
matrix of  the posterior mean of $\m V \m A \m V^T$, obtained by averaging  across samples of the Markov chain.
This produces the matrix in Table \ref{tab:vhats}, which is nearly
identical to the eigenvector matrix of the pooled covariance matrix
$\sum \m Y_k^T \m Y_k/(n_k-1)$, which is also given in the table.
Posterior mean estimates of the eigenvalue matrices $\{ \Lambda_1,\ldots, 
 \Lambda_K\}$  were all within 1.0 of their corresponding 
values based on the 
the empirical covariance matrices. 

Table \ref{tab:eev} suggests that the first and fourth 
eigenvectors are the most preserved across groups, whereas 
the other two are less well-preserved. Letting $\hat {\m  U}_k$ be the 
eigenvector matrix  of $\m Y_k^T\m Y_k$ and $\hat {\m V}$ the eigenvector matrix of 
 $\sum \m Y_k^T \m Y_k/(n_k-1)$, the differential heterogeneity of 
the eigenvectors  can be described numerically
by computing the value of 
  ${\rm diag}( \hat {\m V}^T \hat {\m U}_k )^2$ and averaging each of the $p$ 
entries of this vector across the $K$ groups. This $p$-dimensional 
function of the observed data gives 
$t( \m Y_1^T\m Y_1,\ldots,  \m Y_4^T\m Y_4) = (0.98, 0.85, 0.86, 0.96)$, 
indicating that by this metric the first and last eigenvectors 
are most preserved across groups. 
To examine how well the model represents this observed 
heterogeneity, the value of $t( \m Y_1^T\m Y_1,\ldots,  \m Y_4^T\m Y_4)$  
can be compared to its posterior predictive distribution under the model. 
This  was done by generating simulated values $\tilde {\m Y_1}^T\tilde {\m Y}_1,\ldots, 
  \tilde {\m Y}_4^T\tilde {\m Y}_4$ 
every 10th iteration of the Markov chain and computing the 
statistic $t()$ described above, resulting in 1000 samples of 
$t(\tilde {\m Y_1}^T\tilde {\m Y}_1,\ldots,\tilde {\m Y}_4^T\tilde {\m Y}_4)$
    from the posterior predictive distribution. 
For simplicity we present below only the minimum and maximum 
values of this statistic, which for our observed data are
(0.85, 0.98).

\begin{table}
\begin{center}
\begin{tabular}{rrrrrrrrrr}
\multicolumn{4}{c}{hierarchical model estimate}  & & & \multicolumn{4}{c}{empirical estimate} \\
0.54  &  -0.27  &  -0.19  &  0.77 &&  &  0.55  &  -0.25  &  -0.17  &  0.78  \\
0.54  &  -0.10  &  0.80  &  -0.22 &&  &  0.53  &  -0.10  &  0.81  &  -0.23  \\
0.56  &  -0.15  &  -0.56  &  -0.59&&   &  0.57  &  -0.16  &  -0.56  &  -0.57  \\
0.30  &  0.95  &  -0.06  &  0.10 & & &  0.30  &  0.95  &  -0.05  &  0.08
\end{tabular}
\end{center}
\caption{Model-based and empirical estimates of $\m V$.}
\label{tab:vhats}
\end{table}

The top of the second panel of Figure \ref{fig:voleana} 
shows the posterior predictive distributions of this statistic
on a logit scale 
under the hierarchical model which pools information across 
eigenvector matrices. 
The observed values are well within the predicted range, 
%The plot indicates that the 
indicating that the
model is able to represent the differential amounts of eigenvector preservation among 
the observed covariance matrices. % that is present in the observed data. 
  In contrast, the lower
part of the plot shows a posterior predictive distribution 
generated under a one-shared-eigenvector model similar to the one in 
\citet{flury_1987},  obtained 
obtained by fixing $w=1000$, $\alpha_1=\beta_1=1$ and 
$\alpha_j=\beta_j=0$ for $j>1$. This model accurately predicts 
the highest degree of preservation across eigenvectors, 
but underestimates the preservation among other eigenvectors. 
This is not surprising, as this model shares information 
only across a single eigenvector.

\begin{figure}
\centerline{\includegraphics[height=2.85in]{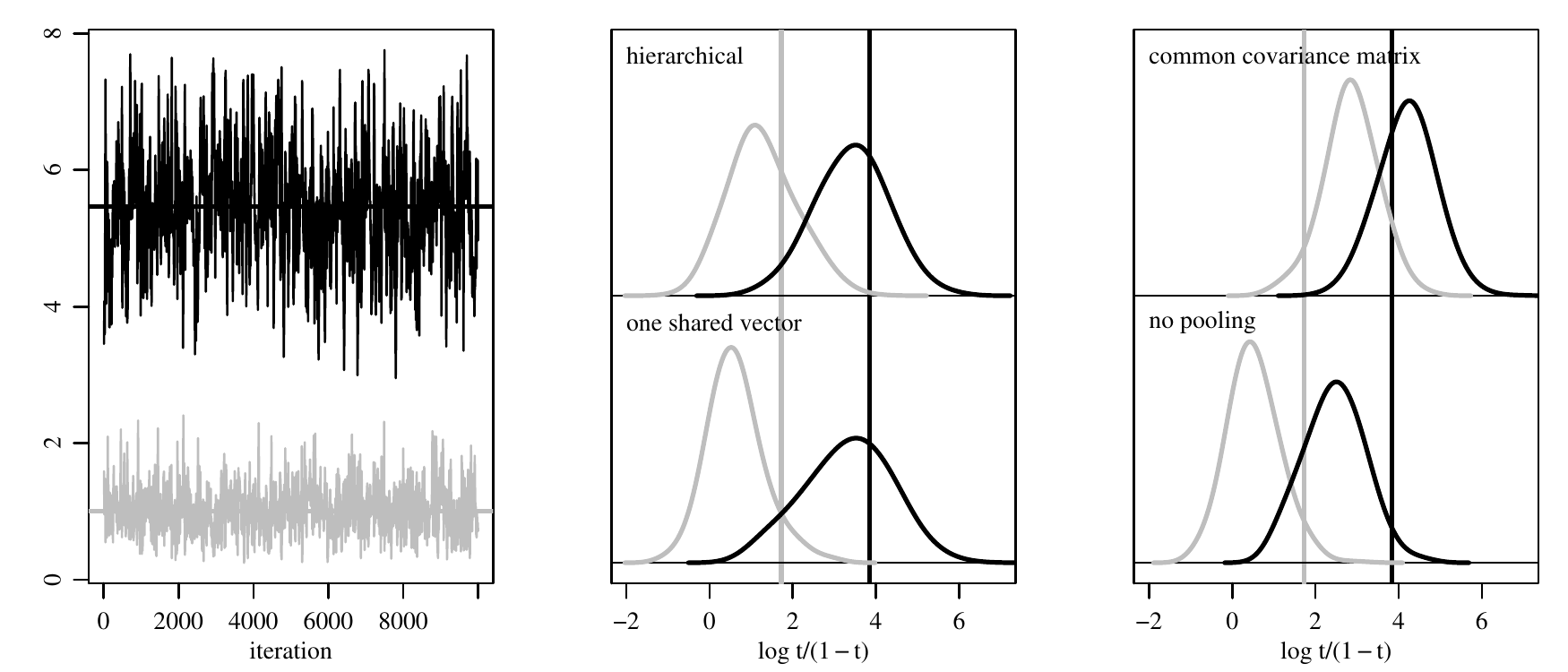}}
\caption{MCMC and model diagnostics for the Vole example.
 The first panel shows
 averages (black) and standard deviations (gray) of the log entries of
  $\m A\circ \m B$
 %$w\v \alpha \v \beta^T$
 at every 10th scan of the Markov chain.
The second  and third panels
show posterior predictive distributions of the minimum and maximum
similarity statistics under a variety of models, with the observed value of each
statistic represented by a vertical line. }
\label{fig:voleana}
\end{figure}

Lastly we fit two other related models:
a ``no pooling'' model in 
which no information was shared across groups and a common covariance 
matrix model in which it is assumed that the covariances are 
exactly identical 
across groups. Not surprisingly these two models  under- and 
over-represent the similarity across eigenvectors of observed 
covariance matrices, as shown in the third panel of Figure \ref{fig:voleana}. 
Taken together, these results indicate that 
assuming complete equality, or completely ignoring similarity, can 
misrepresent the variability of covariance structure across groups.

\section{Example:  National Health Communication Study}
The 2005 Annenberg National Health Communication Survey
(\href{http://anhcs.asc.upenn.edu}{\tt anhcs.asc.upenn.edu})
gathered self-reported health and lifestyle 
 data from 2,989 members of the adult 
U.S.\   population under the age of 65.  Among the variables recorded
were
the following:

\begin{center}
\begin{tabular}{rl}
{\tt state}: & state of residency (including the District of Columbia)\\
{\tt fruitveg}:& typical number of servings of fruit and vegetables per day \\
{\tt exercise}:& typical weekly frequency of exercise   \\
{\tt bmi}: &body mass index  \\
{\tt alcohol}: & number of days in the month consuming five or more alcoholic drinks \\
{\tt smoke}:& typical number of cigarettes smoked per day \\
{\tt age}: & in years \\
{\tt female}: &indicator of being female \\
{\tt income}: &  household income (19 ordered categories) \\
{\tt edu}: &  education level (no degree, high school, 
          some college, Bachelor's degree or higher) 
\end{tabular}
\end{center}
In this section we estimate state-specific correlation matrices 
in a Gaussian copula model for the $p=9$ ordinal variables above. 
More specifically, 
we model the observed data vectors $\v y_1,\ldots, \v y_n$ within a particular 
state as monotone functions of latent Gaussian random variables, so that
\begin{eqnarray*}
\v z_1,\ldots, \v z_n & \sim & \mbox{i.i.d.\ multivariate normal}(\v 0,\Sigma) \\
y_{i,j} &=& g_j(z_{i,j} ). 
\end{eqnarray*}
The non-decreasing functions $\{ g_1,\ldots, g_p\}$ are state-specific
as is the covariance matrix $\Sigma$. 
We compare two models for 
$\Sigma_1,\ldots, \Sigma_K$, the first being one in which 
no information is shared and  $\m U_1,\ldots, \m U_K$ are {\it a priori} 
independent and uniformly distributed on $\mathcal O_p$. The second is the hierarchical 
eigenmodel in which  $\m U_1,\ldots, \m U_K \sim \mbox{i.i.d.}\ p_B(\m U|\m A,\m B,\m V)$, 
with the parameters $\{\m A,\m B,\m V\}$ unknown and estimated from the data, 
using the same prior distributions as in the previous section. 
For both models, the prior distribution on the eigenvalues 
in each group is such that $\{1/\lambda_1 < \cdots < 1/\lambda_p\}$ 
are the order statistics of $p$ independent exponential(1) random variables. 
We note that both of these models ignore the possibility that heterogeneity 
in correlation matrices might be associated with state-specific characteristics 
such as population size or geographic location (although some ad-hoc exploratory 
analyses suggest these effects are small).

Parameter estimation for this hierarchical copula  model can be accomplished by 
iterative sampling of the parameters from their full conditional 
distributions  
as in Section 3 with the latent $\v z$'s taking the roles 
of the observed $\v y$'s, along with iterative sampling of 
the $\v z$'s from their full 
conditional 
distributions (which are constrained normal distributions). 
This latter step is described for a one-group discrete-data copula model in 
\citet{hoff_2007b}.
We note that this is a type of parameter-expanded estimation scheme
\citep{gelman_vandyk_huang_boscardin_2008}
in that the 
scale of each variable $j$ can be represented by both
$g_j$ and $\Sigma_{j,j}$, and so these quantities are not
separately identifiable. 
However, the posterior distribution of $\{ \Sigma_1,\ldots, \Sigma_K\}$ 
induces 
a posterior distribution over state-specific correlation matrices
$\{\m C_1,\ldots, \m C_K$\}, 
which are the parameters of primary interest in copula models. In 
the posterior analysis that follows 
we focus mostly  on comparing the hierarchical and non-hierarchical posterior 
mean estimates of the 
state-specific correlation matrices. 

Markov chains consisting of  105,000 iterations
were constructed  for each of the two models, with results  from the first 5000 iterations being 
discarded to allow for burn-in. 
Parameter values from the remaining iterations 
were saved every 50th 
iteration, leaving 
2000 Monte Carlo samples  for approximating the 
posterior distributions. 
The correlation parameters mixed reasonably well: In the 
hierarchical model,  the
effective sample sizes for 90\% of the parameters was
greater than 500, and for 50\% it was greater than 1200
(effective sample size is an estimate of the 
number of independent samples required to estimate the 
mean to the same precision as with a given autocorrelated sample). 
Mixing of the hierarchical parameters was slower: The first panel 
of Figure \ref{fig:postdiag}
plots the  mean and  standard deviation of 
the logs of the 
64 non-zero values of 
the matrix  $\m A \circ \m B$  at every 50th scan of the Markov 
chain for the hierarchical model. The effective sample sizes for these
two functions of the parameters were both just over 100. 
%copula model also handels missing (at random) data
\begin{figure}
\centerline{\includegraphics[height=3in]{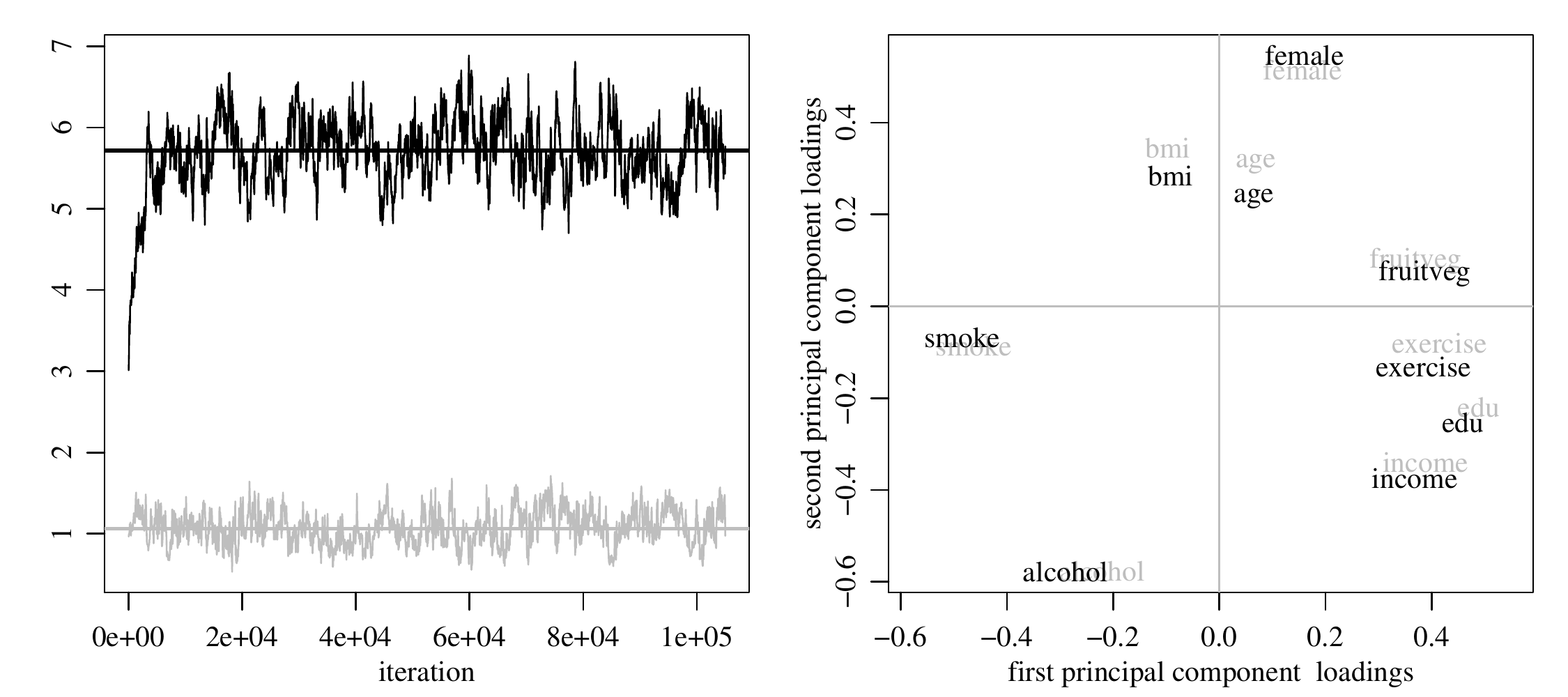}}
\label{fig:postdiag}
\caption{The first panel shows
 averages (black) and standard deviations (gray) of the log entries of
 $\m A \circ \m B$. The second panel shows the first two 
principal component 
loadings for the hierarchical (black) and non-hierarchical (gray) models.}
\end{figure}

\begin{figure}
\centerline{\includegraphics[height=3.90in]{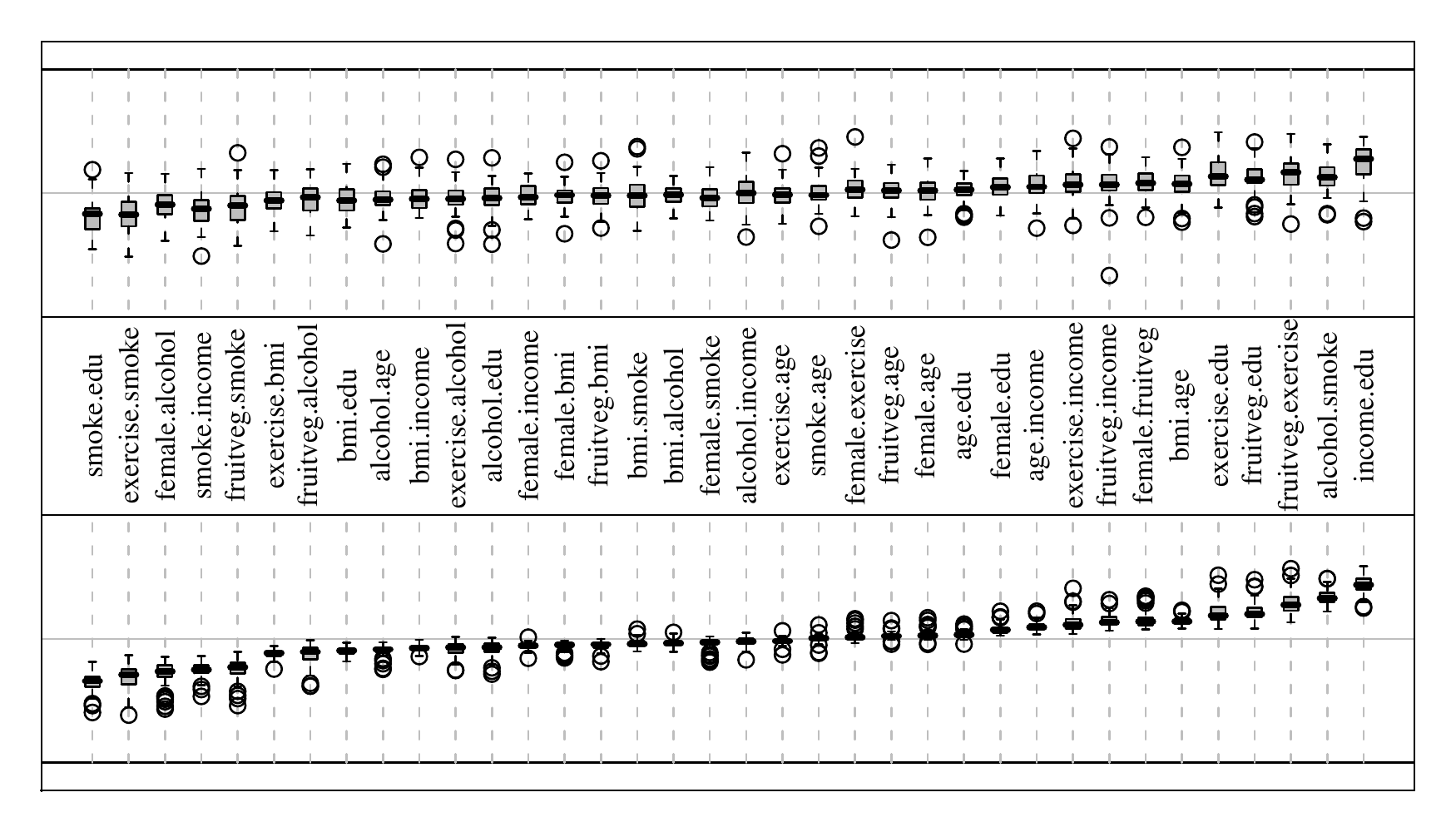}}
\label{fig:corshrink}
\caption{Posterior mean estimates of state-specific correlations. 
The top row are from the non-hierarchical model, the 
bottom from the hierarchical eigenmodel.}
\end{figure}

The second panel of  Figure \ref{fig:postdiag}
plots the first two eigenvectors of the posterior 
mean of the state-averaged correlation matrix 
$\sum_{k=1}^K \m C_k/K$ for each of the two models. 
The results are quite similar, indicating that the 
main correlations across states are described by 
smoking and drinking behavior being negatively 
correlated with education level, income,  fruit and 
vegetable intake and exercise. 
In terms of state-specific correlation matrices however, 
the two models produce quite different results:
The top and bottom plots of Figure \ref{fig:corshrink} give posterior mean
estimates of state-specific correlations from the non-hierarchical and 
hierarchical models, respectively. 
For each pair of variables, a boxplot indicates
the across-state heterogeneity among the $K=51$ parameter estimates
for each  correlation coefficient.
The number of observations per state varies widely, 
with North Dakota and the District of Columbia having 3 respondents
each, whereas Texas and California have 225 and 312 respondents respectively. 
Generally speaking, the estimates for states                 
with lower sample sizes appear at the extremes of the boxplots, 
which is not surprising as these estimates have a 
higher degree of sampling variability. 
% imprecise and  point estimates are subject to a high degree
%of sampling variability. 
Also of  note is the fact that there is no coefficient 
that is estimated as either positive across all states or
 negative across  all states. 
For example, among the highest correlations is that between income 
and education, with an across-state median estimate of 0.28
based on the non-hierarchical model. 
However, there were four states (VT, AK, WY, NE) which 
were estimated as having a negative correlation between these variables. 
The sample sizes from these states were 5, 9, 5 and 15 respectively, 
suggesting that these  low correlations may be due to unrepresentative 
samples. 
In contrast, the hierarchical model recognizes that much of 
the across-state heterogeneity in correlation estimates 
may be due to sampling variability, and shrinks estimates from 
low-sample size states towards the across-state center. 
For example, the hierarchical model gives 
positive point estimates for the correlation between income and education
for all of the states, including VT, AK, WY and NE. 
As shown in the lower half of Figure \ref{fig:corshrink}, across-state heterogeneity 
among  the other correlation coefficients is similarly reduced, 
with 
nearly two-thirds (23 of 36) of the correlation coefficients 
having sign-consistent 
estimates across  all 51 states. 

The effects of hierarchical estimation 
are explored further in 
Figure \ref{fig:eigenshrink}. 
We have two estimates of the eigenvectors for each of 
the $k$  state-specific correlation matrices:
$\hat {\m U}_k$ from the hierarchical model and 
$\check{\m U}_k$ from the  non-hierarchical model. 
We can compute a similarity  between these two 
matrices as the average of the $p$ values of 
 diag$(\check {\m U}_k^T \hat{\m U}_k)^2$. 
The first panel of Figure \ref{fig:eigenshrink} shows that the 
 relationship between the similarity  and the 
within-state sample size
is positive as expected:
Covariance matrices for states with large sample sizes are well-estimated
based on within-state data alone, and  their 
eigenvector estimates are largely unaffected by  hierarchical 
estimation. In contrast,
the amount of information from states with low sample sizes is 
small, and so the estimates for the hierarchical model
 are pulled towards the population 
mode and away from $\check{\m U}_k$. 
The effects of this shrinkage on the  principal axes of 
the correlation matrices are shown  graphically
in the second and third panels of Figure \ref{fig:eigenshrink}. 
The second panel plots the projections of the first 
two columns of each $\check {\m U}_k$ onto the first two columns of the  
eigenvector matrix of the 
pooled correlation matrix. Although heterogeneous, the vectors 
are generally in the same direction, and further inspection shows that 
outliers tend to be states with low sample sizes. The third panel of 
the figure shows the same plot for the projections 
of the columns of each $\hat {\m U}_k$ from the hierarchical model. 
The heterogeneity here represents the estimated across-state 
variability in eigenvectors
after accounting for the within-state sampling variability. 
The axis in this plot that is furthest from the center is that 
representing  Wisconsin, which has relatively high sample size of 69 but 
some extreme correlations:
For example, among states with sample sizes greater than 20, 
Wisconsin has the lowest non-hierarchical estimate of 
the correlation for ({\tt income}, {\tt education}) and 
  ({\tt female}, {\tt bmi}), and the 
 highest  non-hierarchical estimate of
the correlation for ({\tt income}, {\tt alcohol}). 
These correlations make Wisconsin somewhat of an outlier in terms
 of the correlations represented by the first two principal 
components. The relatively large sample size for Wisconsin 
suggests these extreme correlations cannot be solely attributed to 
within-state sampling variability, and this is reflected in the 
state-specific estimated  correlation matrix 
from the hierarchical 
model.

%Under the hierarchical model, Wisconsin had the lowest estimated correlation 
%between {\tt income} and {\tt edu}, 
%{\tt income} and {\tt exercise},  and 
%{\tt income} and {\tt fruitveg}
%and the highest and second highest correlations between 
%{\tt income} and {\tt smoking} and
%{\tt income} and {\tt alcohol}. 
%% need to plot the first two pc's

%The richness of this dataset suggests a number of more 
%detailed analyses. In particular,
%the heterogeneity in correlation matrices could be associated 
%with state-specific characteristics such as 
%population sizes (which are related to sample sizes 
%in this dataset) and 
%geographic locations. 

\begin{figure}
\centerline{\includegraphics[height=2.15in]{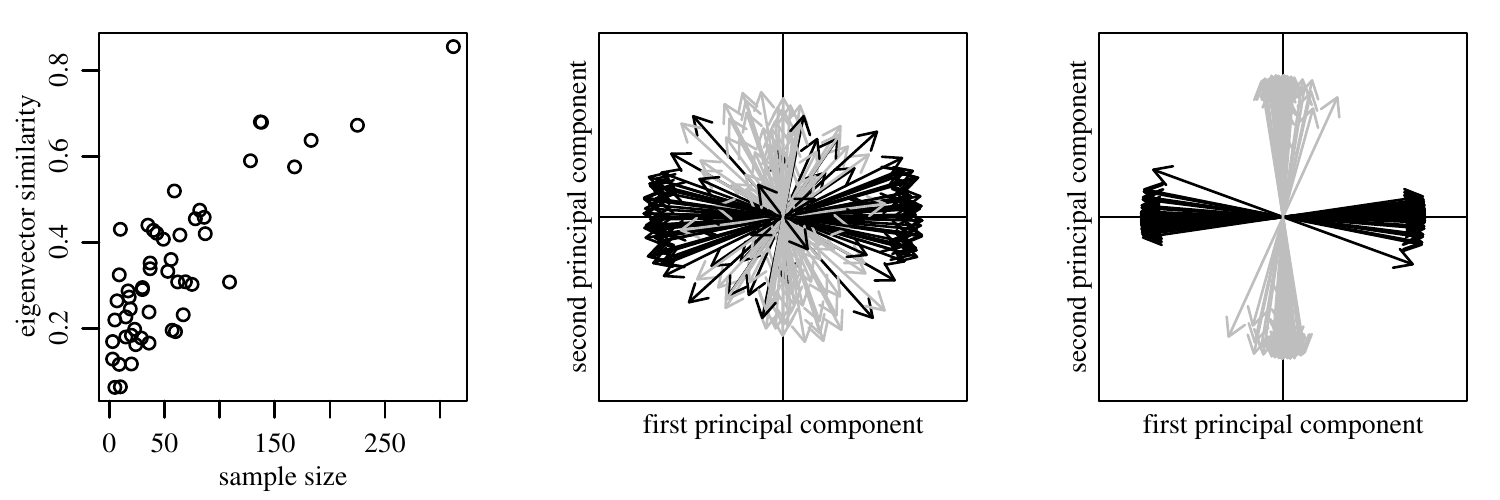}}
\label{fig:eigenshrink}
\caption{Effects of shrinkage on the estimated principal axes. 
The first panel shows the similarity between non-hierarchical  and hierarchical
estimates as  a function of sample size. The second and third show 
heterogeneity across the first two principal axes in the 
non-hierarchical and hierarchical models, respectively. }
\end{figure}

\section{Discussion}
As an alternative to the Bingham distribution, 
a simpler model for across-group covariance heterogeneity 
would be that $\Sigma_1,\ldots, \Sigma_K$ are i.i.d.\ samples from an 
inverse-Wishart distribution. For many applications however, such a model 
may be too simple: The inverse-Wishart distribution has only one parameter
to represent heterogeneity around the mean covariance matrix, and cannot 
represent differential amounts of eigenvector heterogeneity as the Bingham 
distribution can. Additionally, the inverse-Wishart model cannot distinguish 
between across-group eigenvector heterogeneity and across-group 
eigenvalue heterogeneity, as these quantities
are modeled simultaneously.

As we are pooling eigenvector information
across groups it is natural to consider
pooling eigenvalues as well. This would entail
modeling  $\{ \Lambda_1,\ldots, \Lambda_K\}$ as being samples from a common
population, and estimating the parameters of this population using
the data from all $K$ groups. One simple approach to doing this would
be
to estimate the parameters $(\nu_0,\sigma_0^2)$ in the prior 
distribution for the eigenvalues, thus treating the distribution 
as a sampling model. 
As with the other unknown parameters,
this can be done by iteratively updating these parameters based
on their full conditional distributions.
%As a function of $(\nu_0,\sigma_0^2)$, the
%probability density of $\{\Lambda_1,\ldots, \Lambda_K\}$ is
%proportional to
%\begin{eqnarray*}
%p(\Lambda_1,\ldots, \Lambda_K | \nu_0,\sigma_0^2 ) \propto
% p(\nu_0,\sigma_0^2) & \times  & 
%  \Gamma(\nu_0/2)^{pK} \nu_0^{\nu_0 p K/2} \exp\left \{ \frac{\nu_0}{2}\sum_{k}\sum_{j} 
%   \log  \lambda_{j,k}^{-1}  \right \}  \times \\ 
% & & (\sigma_0^2)^{\nu_0 p K/2} \exp \left \{ -\sigma_0^2  \frac{\nu_0}{2} 
%    \sum_k \sum_j \lambda_{j,k}^{-1}  \right \}
%\end{eqnarray*}
%The second line of the above equation indicates 
Straightforward calculations show that a gamma prior
distribution for $\sigma_0^2$ results in a gamma full conditional distribution.
The full conditional distribution for $\nu_0$ is non-standard, but if
$\nu_0$ is restricted to the integers then its
full conditional distribution can easily be sampled from.

Another possible model extension is to situations where the number of variables 
is larger     
than any of the within-group sample sizes. 
In these cases, full-rank covariance estimation can become unstable and computationally 
intractable. 
A remedy to this problem is to use a factor analysis model, 
in which a covariance matrix $\Sigma_k$ is assumed equal to 
$\m U_k \m D_k \m U_k^T + \sigma_k^2\m I$, where $\m D_k$ is 
a positive diagonal 
matrix and
 $\m U_k$ is a $p\times r$ orthonormal matrix with $r<p$, an element of 
the Stiefel manifold $\mathcal S_{r,p}$. 
As before, heterogeneity across covariance matrices can be expressed 
by heterogeneity in these matrix components, and a Bingham model on 
 $\mathcal S_{r,p}$, similar to the one used in 
this paper, can be used to express heterogeneity among $\m U_1,\ldots, \m U_K$. 

Computer code and data for the examples in this article are available 
at 

\begin{center}
\href{http://www.stat.washington.edu/hoff/}{\tt http://www.stat.washington.edu/hoff/}. 
\end{center}

\bibliographystyle{plainnat}
\bibliography{main}

\end{document}